\documentclass[sigconf]{acmart}

\setcopyright{none}
\renewcommand\footnotetextcopyrightpermission[1]{} 
\AtBeginDocument{%
  \providecommand\BibTeX{{%
    \normalfont B\kern-0.5em{\scshape i\kern-0.25em b}\kern-0.8em\TeX}}}

\def\BibTeX{{\rm B\kern-.05em{\sc i\kern-.025em b}\kern-.08em
    T\kern-.1667em\lower.7ex\hbox{E}\kern-.125emX}}
\newtheorem{theorem}{Theorem}[section]
\newtheorem{lemma}[theorem]{Lemma}

\usepackage{subcaption}
\usepackage[normalem]{ulem}
\usepackage[colorinlistoftodos,prependcaption]{todonotes}
\usepackage{xargs}
\newcommandx{\hp}[2][1=]{\todo[inline,linecolor=red,backgroundcolor=red!25,bordercolor=blue,#1]{#2}}
\newcommandx{\zw}[2][1=]{\todo[linecolor=yellow,backgroundcolor=yellow!40,bordercolor=brown,#1]{#2}}
\usepackage{booktabs}
\usepackage{pifont}
\usepackage{hyperref}
\usepackage{glossaries}  
\usepackage{tikzscale}

\usepackage{amsmath,amsfonts,amsthm}
\usepackage{algorithmic}
\usepackage{graphicx}
\usepackage{textcomp}
\usepackage{xcolor}
\usepackage{tikz}
\usetikzlibrary{patterns}

\tikzset{every picture/.style={line width=0.75pt}} 
\usetikzlibrary{shadows.blur}
\usepackage{xstring}
\usepackage{xparse}
\newacronym{cots}{COTS}{commercial off-the-shelf}
\newacronym{LLC}{LLC}{last-level cache}
\newacronym{LRU}{LRU}{least-recently used}
\newacronym{wcl}{WCL}{worst-case latency}
\newacronym{noc}{NoC}{network-on-chip}
\newacronym{msi}{MSI}{Modified-Shared-Invalid}
\newacronym{l1i}{L1I\$}{L1 instruction cache}
\newacronym{l1d}{L1D\$}{L1 data cache}
\newacronym{l2c}{L2\$}{L2 cache}
\newacronym{l3c}{L3\$}{L3 cache}
\newacronym{qlt}{QLT}{Queue Lookup Table}
\newacronym{sqs}{SQ}{Sequencer}

\newcommand{\slot}[2]{\ensuremath{s^{#1}_{#2}}}
\newcommand{\cua}{\ensuremath{c_{ua}}}
\newcommand{\ci}{\ensuremath{c_i}}
\newcommand{\cj}{\ensuremath{c_j}}
\newcommand{\Sw}{\ensuremath{SW}}
\newcommand{\slx}{\ensuremath{set_{LLC}(X)}}
\newcommand{\cl}{\ensuremath{c(l)}}
\newcommand{\clx}{\ensuremath{c(l_x)}}

\ExplSyntaxOn
\NewDocumentCommand { \giveme } { m m }
{
    \clist_item:nn { #1 } { #2 }
}
\ExplSyntaxOff
\newcommand{\dst}[1]{\ensuremath{d^{\giveme{#1}{1}}_{\giveme{#1}{2}}}}
\newcommand{\dstf}[1]{\ensuremath{d^{\giveme{#1}{1}}_{\giveme{#1}{2}}(\giveme{#1}{3})}}

\newcommand{\stua}{\slot{t}{\cua}}
\newcommand{\cir}[1]{\ensuremath{\textcircled{\raisebox{-0.9pt}#1}}}
\newcommand{\dd}{\textsf{distance}}
\newcommand{\cbut}[1]{\ensuremath{C\backslash\{#1\}}}
\newcommand{\rop}{reduce one period}
\newcommand{\mcua}{\ensuremath{m_{\cua}}}
\newcommand{\stsqr}{set sequencer}
\newcommand{\partition}{$\mathcal{P}$}
\newcommand{\onetdm}{1S-TDM}

\theoremstyle{remark}
\newtheorem{observation}{Observation}

\newcommand*\circled[1]{\tikz[baseline=(char.base)]{
            \node[shape=circle,draw,inner sep=0.5pt] (char) {#1};}}

\newcommand{\ssqx}{\texttt{SS}}
\newcommand{\nssx}{\texttt{NSS}}
\newcommand{\ptnx}{\texttt{P}}
            
\newcommand{\ssq}[1]{\texttt{SS(#1)}}
\newcommand{\nss}[1]{\texttt{NSS(#1)}}
\newcommand{\ptn}[1]{\texttt{P(#1)}}

\begin{document}

\title{Predictable Sharing of Last-level Cache Partitions for Multi-core Safety-critical Systems}

\author{Zhuanhao Wu}
\email{zhuanhao.wu@uwaterloo.ca}
\affiliation{%
  \institution{University of Waterloo}
}

\author{Hiren Patel}
\email{hiren.patel@uwaterloo.ca}
\affiliation{%
  \institution{University of Waterloo}
}

\renewcommand{\shortauthors}{Wu and Patel.}

\begin{abstract}

Last-level cache (LLC) partitioning is a technique to provide temporal isolation and low worst-case latency (WCL) bounds when cores access the shared LLC in multicore safety-critical systems. A typical approach to cache partitioning involves allocating a separate partition to a distinct core. A central criticism of this approach is its poor utilization of cache storage. Today’s trend of integrating a larger number of cores exacerbates this issue such that we are forced to consider shared LLC partitions for effective deployments. This work presents an approach to share LLC partitions among multiple cores while being able to provide low WCL bounds. 

\end{abstract}

\maketitle{}
\pagestyle{plain}

\section{Introduction}
The use of multicores in safety-critical systems offers an attractive opportunity to consolidate several functionalities onto a single platform with the benefits of reducing cost, size, weight, and power while delivering high performance. 
Although multicores are mainstay in general-purpose computing, their use in safety-critical systems is approached with caution. 
This is because multicores often share hardware resources to deliver their high performance, but since safety-critical systems must be certified, this makes guaranteeing compliance with safety standards increasingly challenging.
The central reason behind this difficulty is that shared resources complicate worst-case timing analysis necessary for applications deemed to be of high criticality. 
For instance, the automotive domain uses the ISO-26262~\cite{iso26262} standard, which identifies ASIL-D as the highest criticality application where a violation of its temporal behaviours may result in a significant loss of lives or injury. 

One such shared hardware resource is the shared last-level cache (LLC) that multiple cores access when they experience misses in their private caches. 
For example, the Kalray MPPA 3~\cite{kalraymppa2}  features an 80-core architecture with 16 cores in a cluster that share 4MB of LLC.
LLCs are an important component of the memory hierarchy to deliver high performance~\cite{navarro2019memory}. 
However, multiple cores accessing the LLC can introduce inter-core temporal interferences where one core evicts the data of another's resulting in large variations in execution times.
These interferences complicate worst-case latency (WCL) analysis, and often result in overly pessimistic worst-case bounds.
LLC partitioning has been proposed as a countermeasure to address these difficulties in using LLCs with multicores~\cite{cacheparteval,gracioli2015,lv2016}. 
LLC partitioning allocates a part of the LLC to each core that it can use. 
This provides temporal isolation to tasks executing on a core from other tasks executing on another core. 
However, there are multiple downsides to LLC partitioning: (1) it can significantly affect average-case performance, (2) it can lead to underutilization of cache capacity, and (3) prevent coherent data sharing~\cite{pred1}.
Downside (1) is a result of each core having a smaller part of the LLC. 
(2) happens when a core gets allocated a partition that it doesn't effectively use.
Lastly, for (3), conventional LLC partitioning disallows one core to access a partition of another core; thus, accessing shared data between cores in the LLC is prohibited. 
This prevents LLC caching of coherent data across multiple cores.
With the continued increase in demand for functionalities, and their consolidation onto a multicore platform, we expect these downsides to overwhelm the benefits of LLC partitioning, and force us to seriously consider sharing LLC partitions in the near future.

As a cautious step towards addressing these downsides, and possibly a refreshing alternative to traditional LLC partitioning approaches, we allow multiple cores to share partitions.
This requires us to determine the WCL of memory accesses from cores that miss in their private caches, and access the shared partition.
In this paper, we develop such a WCL analysis.
In doing so, we show that naively arbitrating cores' accesses to the shared LLC partitions results in a scenario where the WCL is unbounded.
We correct this unbounded scenario by showing that a restricted version of time-division multiplexing (TDM) policy called \onetdm{} can result in a WCL bound.
However, the resulting WCL bound is grossly pessimistic; it is proportional to the minimum of the cache capacity and LLC partition size of a given core and cube of the number of cores.
By methodically analyzing the critical instance that renders the WCL bound, we intuit a technique to significantly lower the WCL.
This technique yields a WCL that eliminates the dependency on the cache and partition sizes. 
For a 4-core setup with a 16-way LLC with 128 cache lines, our approach results in a WCL that is 2048 times lower.
We implement this technique in a hardware structure called the \stsqr{}.
We also show that careful sharing of cache partitions not only allows for a low WCL, but possibly higher average-case performance.
We envision the proposed work to complement existing efforts on LLC partitioning where certain tasks have their own partitions, but others share partitions; all of which depends on their performance and real-time requirements.
The following are our main contributions.
  \begin{itemize}
  \item We identify that naively using TDM to arbitrate accesses to the shared LLC partitions can result in an unbounded WCL. 
  We resolve this by showing that a \onetdm{} schedule prohibits this scenario.
    \item We develop a WCL analysis for a memory access to a shared LLC partition using the \onetdm{} arbitration policy.
    \item We propose a micro-architectural extension called the \stsqr{} that significantly lowers the WCL when sharing a LLC partition.
    \item We evaluate the proposed approach by implementing a simulation of the \stsqr{}.
  \end{itemize}

\section{Related works}
Cache partitioning~\cite{gracioli2015} reserves a portion of the cache to a task or a core either via hardware or software techniques~\cite{wang2017}.
The key role of cache partitioning is to improve temporal isolation to simplify the WCL analysis. 
However, as the number of cores increase, allocating distinct partitions to each core or each task can affect average-case performance and cache under-utilization.
Moreover, this prohibits deploying a large number of functionalities on the multicore as it may result in extremely small partitions to each functionality, which would adversely affect performance.
The proposed work seeks a middle ground where designers can judiciously share partitions with a subset of cores, and isolate others.
Prior works identified that accurately capturing  the contention of multiple cores sharing the LLC is  difficult~\cite{suhendra2008}, and attempts exist for shared cache for dual-core processor relying on knowledge of the application~\cite{zhang2012}.
Our work does not rely on application-specific knowledge and does not constraint the number of cores.

Authors in~\cite{cerrolaza20} noticed that interference between tasks exists when multiple tasks or cores share the LLC.
They proposed a time analysable shared LLC where the inter-task interference is bounded probabilistically.
Given their proposed technique, the exact worst-case interference in the LLC and the exact worst-case latency is not discernible.
Compared to \cite{cerrolaza20}, our analysis provides an exact (non-probabilistic) bound of a memory access in the LLC, and does not rely on MBPTA.
Our work assumes that a TDM bus arbitration has one slot in each period, which is common in controlling access to resources in safety-critical systems~\cite{pred1,rihani2015}.
We also identified the worst-case scenario for the LLC evictions, and that in the worst-case the latency can be unbounded if the arbitration has no constraint.
Moreover, our analysis does not rely on certain type of address mapping or replacement policy.
\section{System Model}
\label{sec:sysmodel}
\begin{figure}[t]
    \centering
    \resizebox{\columnwidth}{!}{\input{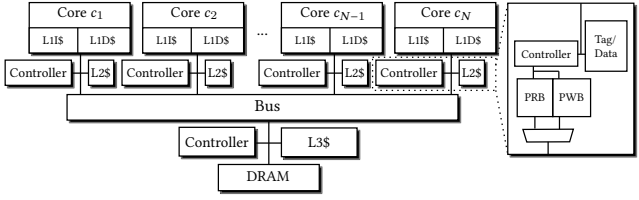}}
    \caption{Overview of system model.}\label{fig:system-model}
\end{figure}
\noindent\textbf{System hierarchy.}
We assume a system with $N$ cores and three levels of cache in the memory hierarchy (Figure~\ref{fig:system-model}).
Each core has a private \Gls{l1i} and a private \Gls{l1d}.
The private L1 caches of a core are connected to a \Gls{l2c}.
A set-associative \gls{LLC}, or the \Gls{l3c}, connects to all L2 caches, and interfaces with a DRAM directly.
The accesses from L2 to L3 is controlled by a shared bus.
The shared bus implements a time-division multiplexing (TDM) arbitration where equally sized slots are allocated to cores.
The L2 cache controller of each core only accesses the \gls{LLC} in the core's allocated slot, and the \gls{LLC} only responds to the L2 within the core's slot.
Since the L2 cache is private to each core, a core accessing the bus implies that the core's L2 cache controller access the bus.
An LLC partition isolates a part of the LLC for a specific core. 
We assume one task can be mapped to one core. 

Similar to prior works~\cite{pred1}, we assume each core has at most one outstanding memory request.
Before a core's request or write-back is placed on the bus, we assume that the request is buffered in a structure named pending request buffer (PRB), and the write-backs are buffered in a structure named pending write-back buffer (PWB). 
There is a predictable arbitration such as round-robin between PRB and PWB to choose from a request or a write-back to send on the bus at the beginning of the core's slot.
\newline
\noindent\textbf{Cache inclusion policy.}
We assume that the \gls{LLC} is inclusive of L2. 
This is a common setup in existing platforms. 
Suppose that core $c_{i}$'s request to the cache line at address $X$ is a miss in all its private caches, and the LLC. 
Then, for the LLC to respond with the provided data for $X$, the LLC must ensure the following: (1) that there is an empty line in the set that cache line $X$ maps to, (2) the cache line from a lower level memory in the memory hierarchy is fetched, and (3) the response to $c_i$ is sent in $c_i$'s slot.
An important property of inclusive caches is that an eviction of a cache line in the lower-level cache requires eviction of cache lines for the same address in upper-level caches. 
For our setup, an eviction in the LLC, would force evictions in both the L1 and L2 private caches that have the data.

\section{WCL with shared partitions}
\subsection{An unbounded WCL scenario}

\begin{figure}[tb]
    \centering
    \resizebox{\columnwidth}{!}{\tikzset{every picture/.style={line width=0.75pt}} 

\begin{tikzpicture}[x=0.75pt,y=0.75pt,yscale=-1,xscale=1]

\draw  [fill={rgb, 255:red, 255; green, 255; blue, 255 }  ,fill opacity=1 ] (540,120.39) -- (570,120.39) -- (570,140) -- (540,140) -- cycle ;
\draw  [fill={rgb, 255:red, 255; green, 255; blue, 255 }  ,fill opacity=1 ] (330,120.39) -- (360,120.39) -- (360,140) -- (330,140) -- cycle ;
\draw  [fill={rgb, 255:red, 169; green, 196; blue, 215 }  ,fill opacity=1 ] (150,90) -- (210,90) -- (210,110) -- (150,110) -- cycle ;
\draw  [fill={rgb, 255:red, 255; green, 255; blue, 255 }  ,fill opacity=1 ] (150,120.39) -- (210,120.39) -- (210,140) -- (150,140) -- cycle ;
\draw  [dash pattern={on 0.84pt off 2.51pt}]  (210,60) -- (210,190) ;
\draw  [dash pattern={on 0.84pt off 2.51pt}]  (420,60) -- (420,190) ;
\draw    (150,60) -- (150,190) ;
\draw  [dash pattern={on 0.84pt off 2.51pt}]  (270,60) -- (270,190) ;
\draw    (360,60) -- (360,190) ;
\draw  [fill={rgb, 255:red, 194; green, 218; blue, 159 }  ,fill opacity=1 ] (210,120.39) -- (270,120.39) -- (270,140) -- (210,140) -- cycle ;
\draw  [dash pattern={on 0.84pt off 2.51pt}]  (330,60) -- (330,190) ;
\draw  [dash pattern={on 0.84pt off 2.51pt}]  (480,60) -- (480,190) ;
\draw  [dash pattern={on 0.84pt off 2.51pt}]  (540,60) -- (540,190) ;
\draw    (570,60) -- (570,190) ;
\draw  [fill={rgb, 255:red, 255; green, 255; blue, 255 }  ,fill opacity=1 ] (270,120) -- (330,120) -- (330,140) -- (270,140) -- cycle ;
\draw  [fill={rgb, 255:red, 255; green, 255; blue, 255 }  ,fill opacity=1 ] (480,120.39) -- (540,120.39) -- (540,140) -- (480,140) -- cycle ;
\draw  [fill={rgb, 255:red, 255; green, 255; blue, 255 }  ,fill opacity=1 ] (360,120.39) -- (420,120.39) -- (420,140) -- (360,140) -- cycle ;
\draw  [fill={rgb, 255:red, 0; green, 0; blue, 0 }  ,fill opacity=1 ] (326.45,134.36) -- (330.25,140.04) -- (333.94,134.27) -- (331.43,134.3) -- (331.17,110.42) -- (328.69,110.44) -- (328.95,134.33) -- cycle ;
\draw  [fill={rgb, 255:red, 0; green, 0; blue, 0 }  ,fill opacity=1 ] (153.99,86.11) -- (150.25,80.39) -- (146.5,86.11) -- (149,86.11) -- (149,110) -- (151.49,110) -- (151.49,86.11) -- cycle ;
\draw  [fill={rgb, 255:red, 0; green, 0; blue, 0 }  ,fill opacity=1 ] (535.94,134.21) -- (539.75,139.96) -- (543.43,134.12) -- (540.92,134.15) -- (540.66,109.99) -- (538.17,110.02) -- (538.44,134.18) -- cycle ;
\draw  [fill={rgb, 255:red, 194; green, 218; blue, 159 }  ,fill opacity=1 ] (420,120.39) -- (480,120.39) -- (480,140) -- (420,140) -- cycle ;
\draw  [fill={rgb, 255:red, 232; green, 193; blue, 182 }  ,fill opacity=0.75 ] (210,90) -- (570,90) -- (570,110) -- (210,110) -- cycle ;
\draw  [fill={rgb, 255:red, 0; green, 0; blue, 0 }  ,fill opacity=1 ] (483.99,115.8) -- (480.24,110) -- (476.5,115.8) -- (479,115.8) -- (479,140) -- (481.49,140) -- (481.49,115.8) -- cycle ;

\draw (139,97.5) node [anchor=east] [inner sep=0.75pt]   [align=left] {$\displaystyle c_{ua}$};
\draw (151,44.6) node [anchor=north west][inner sep=0.75pt]  [font=\large] [align=left] {$\displaystyle s_{c_{ua}}^{t}$};
\draw (361,44.6) node [anchor=north west][inner sep=0.75pt]  [font=\large] [align=left] {$\displaystyle s_{c_{ua}}^{t+1}$};
\draw (139,127.5) node [anchor=east] [inner sep=0.75pt]   [align=left] {$\displaystyle c_{i}$};
\draw (180,100) node  [font=\large] [align=left] {Req X};
\draw (139,135.4) node [rotate=90,anchor=south east] [inner sep=0.75pt]  [font=\large] [align=left] {$\displaystyle set_{LLC}( X)$};
\draw (240,130.19) node  [font=\large] [align=left] {WB $\displaystyle l_{1}$};
\draw (300,130) node  [font=\large] [align=left] {Req X};
\draw (181,129.5) node  [font=\large] [align=left] {Evict $\displaystyle l_{1}$};
\draw (510,130.19) node  [font=\large] [align=left] {Req X};
\draw (390,129.5) node  [font=\large] [align=left] {Evict $\displaystyle l_{1}$};
\draw (151,152.5) node [anchor=west] [inner sep=0.75pt]  [font=\large] [align=left] {$\displaystyle l_{1} :c_{i}$};
\draw (151,167.5) node [anchor=west] [inner sep=0.75pt]  [font=\large] [align=left] {$\displaystyle l_{2} :c_{i}$};
\draw (151,182.5) node [anchor=west] [inner sep=0.75pt]  [font=\large] [align=left] {$\displaystyle l_{3} :c_{i}$};
\draw (210.92,152.5) node [anchor=west] [inner sep=0.75pt]  [font=\large] [align=left] {$\displaystyle l_{1} :-$};
\draw (210.92,167.5) node [anchor=west] [inner sep=0.75pt]  [font=\large] [align=left] {$\displaystyle l_{2} :c_{i}$};
\draw (210.92,182.5) node [anchor=west] [inner sep=0.75pt]  [font=\large] [align=left] {$\displaystyle l_{3} :c_{i}$};
\draw (270.92,152.5) node [anchor=west] [inner sep=0.75pt]  [font=\large] [align=left] {$\displaystyle l_{1} :c_{i}$};
\draw (270.92,167.5) node [anchor=west] [inner sep=0.75pt]  [font=\large] [align=left] {$\displaystyle l_{2} :c_{i}$};
\draw (270.92,182.5) node [anchor=west] [inner sep=0.75pt]  [font=\large] [align=left] {$\displaystyle l_{3} :c_{i}$};
\draw (360.92,152.5) node [anchor=west] [inner sep=0.75pt]  [font=\large] [align=left] {$\displaystyle l_{1} :c_{i}$};
\draw (360.92,167.5) node [anchor=west] [inner sep=0.75pt]  [font=\large] [align=left] {$\displaystyle l_{2} :c_{i}$};
\draw (360.92,182.5) node [anchor=west] [inner sep=0.75pt]  [font=\large] [align=left] {$\displaystyle l_{3} :c_{i}$};
\draw (420.92,152.5) node [anchor=west] [inner sep=0.75pt]  [font=\large] [align=left] {$\displaystyle l_{1} :-$};
\draw (420.92,167.5) node [anchor=west] [inner sep=0.75pt]  [font=\large] [align=left] {$\displaystyle l_{2} :c_{i}$};
\draw (420.92,182.5) node [anchor=west] [inner sep=0.75pt]  [font=\large] [align=left] {$\displaystyle l_{3} :c_{i}$};
\draw (480.92,152.5) node [anchor=west] [inner sep=0.75pt]  [font=\large] [align=left] {$\displaystyle l_{1} :c_{i}$};
\draw (480.92,167.5) node [anchor=west] [inner sep=0.75pt]  [font=\large] [align=left] {$\displaystyle l_{2} :c_{i}$};
\draw (480.92,182.5) node [anchor=west] [inner sep=0.75pt]  [font=\large] [align=left] {$\displaystyle l_{3} :c_{i}$};
\draw (345.5,76.5) node  [font=\large] [align=left] {...};
\draw (345,130.19) node  [font=\large] [align=left] {...};
\draw (554.5,77.5) node  [font=\large] [align=left] {...};
\draw (555,130.19) node  [font=\large] [align=left] {...};
\draw (179.5,79) node  [font=\large] [align=left] {$\displaystyle \textcircled{1} :c_{ua}$};
\draw (239.5,79) node  [font=\large] [align=left] {$\displaystyle \textcircled{2} :c_{i}$};
\draw (299.5,79) node  [font=\large] [align=left] {$\displaystyle \textcircled{3} :c_{i}$};
\draw (390.5,79) node  [font=\large] [align=left] {$\displaystyle \textcircled{4} :c_{ua}$};
\draw (450.5,79) node  [font=\large] [align=left] {$\displaystyle \textcircled{5} :c_{i}$};
\draw (510.5,79) node  [font=\large] [align=left] {$\displaystyle \textcircled{6} :c_{i}$};
\draw (450,130.19) node  [font=\large] [align=left] {WB $\displaystyle l_{1}$};

\end{tikzpicture}}
    \caption{\cua{} experiences an unbounded latency.}
    \label{fig:motivational-unbounded}
\end{figure}
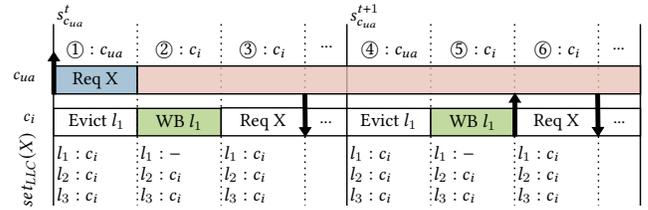

We show that an undesired consequence of the inclusive property, and multiple cores making accesses to the LLC is a situation where the WCL is unbounded. 
Using Figure~\ref{fig:motivational-unbounded}, we illustrate this scenario.
We assume a TDM arbitration policy with one slot for \cua{} and two slots for $c_i$.
We use \slot{t}{\ci} to denote the starting time of the $t$-th slot of core \ci{}.
Unambiguously, \slot{t}{\ci} also refers to the $t$-th slot of \ci{}.
At \stua{}, in \circled{1}, \cua{}'s request to a cache line $X$ misses in the private caches and the LLC.
Cache line $X$ is mapped to a full cache set \slx{} in the LLC; thus, the LLC evicts a cache line $l_1$ in \slx{}, which is also privately cached by $c_i$ denoted by $l_1 :c_i$. 
Note that a full cache set means that there are no empty cache lines in the set.
In \circled{2}, $c_i$ writes back $l_1$, and frees an entry in \slx{}.
Then, in \circled{3}, $c_i$ requests a cache line mapped to \slx{}, and gets the response within its slot because $l_1$ is available.
In \cua{}'s next slot $\slot{t+1}{\cua}$ (\circled{4}), \slx{} is full again preventing \cua{} from completing its request. 
This behavior can potentially repeat indefinitely, which shows that the interference at the LLC from other cores can cause unbounded WCL for the core under analysis.

\subsection{One-slot TDM schedule}
The scenario where \cua{} experiences unbounded latency happens when a core other than \cua{} is allowed to access the LLC multiple times before \cua{} can access the bus again.
An effective solution to prevent another core from occupying a free entry in \slx{}  is to constrain the TDM schedule to allocate only one slot per core in a period (Definition~\ref{def:one-slot}). 
Note that the crux of the issue is that an entry freed by \ci{} due to the eviction of cache line $l_1$ is occupied again by \ci{} again \textit{before} \cua{} can access the LLC. 
With \onetdm{}, we only allow one core to perform one access to the bus in a period.

\begin{definition}\label{def:one-slot}
A \textit{one-slot TDM schedule} (\onetdm) is a TDM schedule that has exactly one slot allocated to each core in every period.   
\end{definition}


\subsection{Key observations}
Although a \onetdm{} schedule guarantees a WCL bound for \cua{}, we show that the WCL is proportional to the minimum of the cache capacity of \cua{} and \cua{}'s LLC partition size $M$, and cube of the number of cores.
This results in a significantly large WCL. 
We illustrate this by making observations from two examples.


\begin{figure}[b]
    \centering
    \resizebox{\columnwidth}{!}{\input{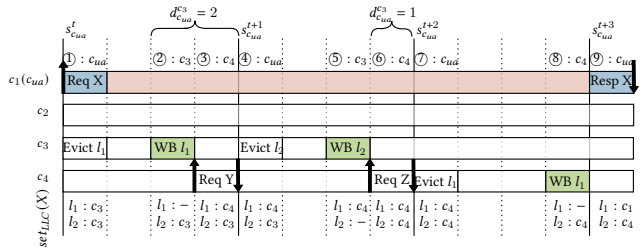}}
    \caption{\cua{}'s request to cache line $X$ eventually is completes.}
    \label{fig:motivational-distance-decrease}
\end{figure}
Consider the example in Figure~\ref{fig:motivational-distance-decrease} that has four cores, and a two-way set-associative LLC.
The \onetdm{} schedule is $\{\cua{},c_2,c_3,c_4\}$.
In \circled{1}, $\cua{}$ requests cache line $X$, which is not privately cached by \cua{}. Hence, its L2 cache controller issues request for $X$ at the beginning of \slot{t}{\cua} to the LLC.
Cache line $X$ maps to a cache set \slx{} in the LLC, and it also experiences a miss in the LLC.
Since \slx{} has no empty lines, it must evict one cache line in \slx{}. 
Suppose that $l_1 \in \slx{}$ is selected for eviction. 
Note that $l_1$ is privately cached in $c_3$ denoted by $l_1:c_3$. 
Hence, $c_3$ must also evict $l_1$ from its private caches in \circled{2}, resulting in a free entry in \slx{}, denoted as $l_1:-$.
Clearly, the cache line to replace depends on the replacement policy.
In this work, we assume a replacement policy that can select any of the cache lines. 
Even though we select $l_1$ in this example, our observation is agnostic of replacement policy.
As a result, the observation applies to any replacement policy, including \gls{LRU}.

Next, in \circled{3}, $c_4$'s L2 cache controller issues a request, and occupies the free entry, preventing \cua{} from obtaining its response in slot $\slot{t+1}{\cua}$.
As a result, the LLC must evict a cache line in \slx{} to make space for \cua{}, and the LLC evicts $l_2$ in slot \slot{t+1}{\cua} which is privately cached by $c_3$.
As before, $c_3$ evicts $l_2$ in \circled{5}, but it is occupied by $c_4$ in \circled{6}.
Note that in the period starting at \slot{t+2}{\cua{}}, this pattern of interfering \cua{} from receiving its response cannot continue.
When $c_4$ evicts $l_1$ in \circled{8}, there is no other core that can occupy the free entry; thus, \cua{} will get its response in \slot{t+3}{\cua{}} at \circled{9}.
This is guaranteed to occur because whenever any core other than \cua{} occupies a free entry in \slx{}, \cua{} gets \textit{closer} to being able to occupy a free entry in \slx{}. 
For example, in \circled{7}, both cache lines are privately cached by $c_4$, and any core making a request to \slx{} resulting in a miss requires $c_4$ to evict it from its private caches as well. 
Since, \cua{}'s slot comes after $c_4$'s, it is guaranteed to occupy the free cache line entry due to $c_4$'s eviction. 
We characterize this \textit{closer} effect by introducing a notion of \dd{} (Definition~\ref{def:dist}).
%
\begin{definition}\label{def:dist}
For a \onetdm{} schedule $S$, the \dd{} between two cores \ci{} and \cj{}, \dst{\ci,\cj}, is the number of slots between the start of slot of \ci{}, and the start of \cj{}'s next slot.

\end{definition}

\begin{corollary}\label{corollary:dist-range}
Given a \onetdm{} schedule $S$ with $N$ cores, the \dd{} between any two cores \ci{} and \cj{}, \dst{\ci,\cj}, $1 \leq \dst{\ci,\cj} \leq N$.
\end{corollary}
Given a cache line $l$, we will use $\dstf{\cl,\cj,x}$ as a convenience to return the \dd{} between the core that has privately cached $l$, and  \cj{} at $x$. 
Using Figure~\ref{fig:motivational-distance-decrease}, with a TDM schedule of $\{\cua{}, c_2, c_3, c_4\}$, $\dst{c_3,\cua{}} = 2$ and $\dst{c_4,\cua{}} = 1$. 
With Definition~\ref{def:dist}, the example in Figure~\ref{fig:motivational-distance-decrease} can be interpreted in terms of \dd{}:
the core that caches $l_1$ changes from $c_3$ in slot \cir{1}, with a \dd{} of $\dst{c_3, \cua}=2$, to $c_4$ in slot \cir{4}, with a \dd{} of $\dst{c_4,\cua}=1$, and finally freed in slot \circled{8}.
Similarly, the core that caches $l_2$ changes from $c_3$ in slot \circled{1} to $c_4$ in slot \circled{7}, and thus the \dd{} of the core that caches $l_2$ decreases from $2$ to $1$.
These example scenarios highlight the following key observations. 
\begin{observation}\label{obs:1}
Given a \onetdm{} schedule $S$, the \textit{\dd{}} of cache lines in \slx{} decrease when \cua{} does not perform  write-backs after issuing its request to cache line $X$ and before receiving its response for $X$.
\end{observation}

The decreasing \dd{} articulates the effect of the core under analysis getting \textit{closer} to occupying a freed cache line entry in \slx{}. 
A direct consequence of observation~\ref{obs:1} is that \cua{}'s request will eventually complete as described in observation~\ref{obs:2}.
\begin{observation}\label{obs:2}
\cua{}'s request will eventually complete. 
\end{observation}
The main intuition behind this observation is that once the lines in \slx{} are privately cached by $c_4$, a request for $X$ by \cua{} will succeed in the following period (Figure~\ref{fig:motivational-distance-decrease}). 
This is because $c_4$ must evict the privately cached line due to inclusive property, which results in a free entry in \slx{} that \cua{} can occupy.
When $n \leq N$ cores share a partition with a \onetdm{} schedule and there are $w$ ways in \slx{}, for \cua{} to occupy an entry in \slx{} in the worst-case, the \dd{} of all $w$ cache lines must experience the largest decrements.
Since the maximal \dd{} is $n$ when \cua{} caches a cache line, and the minimal \dd{} is $1$, \cua{} must wait for the \dd{} of all $w$ cache lines to decrease from $n$ to $1$, accounting for $w(n-1)$ decrements in the worst-case.

Note that we have not considered scenarios where \cua{} performs write-backs before receiving the response for its request. 
When a core other than \cua{} requests a cache line that is privately cached by \cua{}, then \cua{} would also need to perform write-backs due to inclusivity. 
The effect of write-backs on the \dd{} is summarized in observation~\ref{obs:3}.
After a write-back by \cua{}, cache lines in \slx{} can be privately cached by a core with a larger \dd{} compared to before \cua{} performs the write-back. 
\begin{observation}\label{obs:3}
Given a \onetdm{}, when \cua{} performs a write-back after issuing its request to a cache line $X$ and before receiving the response, the \dd{} of cache lines in \slx{} increases. 
\end{observation}
\begin{figure}[tb]
    \centering
    \resizebox{\columnwidth}{!}{\input{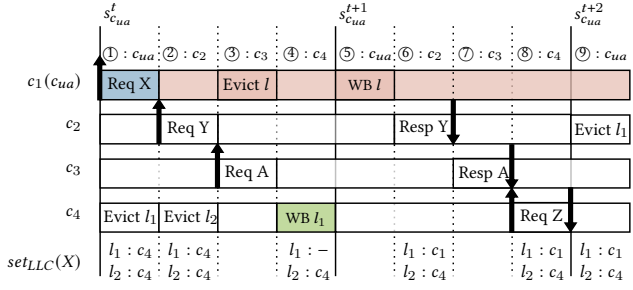}}
    \caption{Distance of core caching $l_1$ increases after \slot{t+1}{\cua}.}
    \label{fig:motivational-distance-increase}
\end{figure}
Figure~\ref{fig:motivational-distance-increase} shows a scenario with \cua{} performing write-backs. 
There are four cores with \onetdm{} schedule of $\{\cua{}, c_2, c_3, c_4\}$.
Cache lines $l_1$ and $l_2$ are in \slx{}, and are initially privately cached by $c_4$.
In \circled{1}, \cua{} issues a request to cache line $X$ that misses in the private caches and LLC.
This is followed by $c_2$ issuing a request to cache line $Y$ such that $Y \in \slx{}$ in \circled{2}, and it also misses in the private caches and LLC. 
The LLC selects to evict another line $l_2$, which needs to be evicted by $c_4$. 
In \circled{3}, $c_3$ issues a request to cache line $A$ that causes $c_1$ to evict a cache line $l$.
In \circled{4}, $c_4$ writes back $l_1$ freeing up an entry in \slx{}.
In the write-back slot of \slot{c_1}{t+1} \circled{5}, $c_1$'s request to $X$ cannot be satisfied because $c_1$ has to perform an eviction.
The free entry is thus occupied by $c_2$.
Note that the core that is caching $l_1$ has changed from $c_4$ to $c_2$ and thus the \dd{} of the core caching $l_1$ increased from $\dst{c_4,c_1}=1$ to $\dst{c_2,c_1}=3$.
In general, write-backs from \cua{} allow a core with a larger \dd{} to occupy a free entry in \slx{} that would have satisfied \cua{}'s request; thus, the \dd{} of cores caching cache lines in \slx{} does not always decrease as in the case of Observation 1 when \cua{} performs write-backs.
Combining these two observations, we develop an analysis that bounds the worst-case latency of a request.

\subsection{WCL analysis for \onetdm{} schedule}

We first prove that the \dd{} in \slx{} only decreases when no write-back by \cua{} is involved with Corollary~\ref{corollary:decreasing-3} (Observation 1).
Then, we bound the latency required for the \dd{} to decrease.
Next, in Lemma~\ref{lemma:write-back}, we show that \dd{} increases when \cua{} writes back cache lines (Observation 2).
Hence, when \cua{} waits for its response, the \dd{} in \slx{} shows an alternating pattern of decreasing and increasing.
Finally, Theorem~\ref{thm:bound-simple} combines Corollary~\ref{corollary:decreasing-3} and Lemma~\ref{lemma:write-back} to express the WCL of \cua{}'s request.

We use our key observations to formulate an analysis to compute the WCL.
Consider a multicore configuration with $N$ cores interacting over the shared bus using a \onetdm{} schedule $S$, and $n$ cores sharing a partition \partition{} in the LLC ($n \leq N$) with \cua{} being one of the $n$ cores.
Throughout the analysis, we assume that \cua{}'s request for cache line $X$ misses in its private caches and the LLC, and \slx{} is full before \cua{}'s request to cache line $X$ is completed.

\begin{lemma}\label{lemma:decreasing-2}
If \cua{}'s request is not completed at slot \slot{t+T}{\cua}, \cua{} does not perform any write-backs, and $l_x\in\slx$ is evicted in response to \cua{}'s request in \slot{t}{\cua}, then
\begin{align*}
    \forall l\in\slx : \ \dst{\cl,\cua} (\slot{t+T}{\cua}) \leq \dst{\cl,\cua}(\slot{t}{\cua}).
\end{align*}

\end{lemma}

\begin{proof}

We prove the lemma by contradiction and assume that $\exists l\in\slx{}: \dst{\cl,\cua} (\slot{t+T}{\cua}) > \dst{\cl,\cua}(\slot{t}{\cua})$.
Then, before \slot{t+T}{\cua}, there must exist two cores \ci{} and \cj{}, such that \ci{} frees the entry $l$ and \cj{} occupies $l$ after \ci{} frees $l$.
The freeing-then-occupying by \ci{} and \cj{} increases the \dd{} of $l$ to be greater than $\dst{\clx,\cua}(\stua)$, that is, $\dst{\ci,\cua}\leq\dst{\clx,\cua}(\stua)<\dst{\cj,\cua}$.
Assume that $\ci{}$ frees $l$ in $\slot{q}{\ci}$ and $\cj{}$ occupies $l$ in $\slot{r}{\cj}$, then $\slot{q}{\ci} < \slot{r}{\cj}$.
Furthermore, because $\dst{\ci,\cua} < \dst{\cj,\cua}$ and \onetdm{} is deployed, \slot{r}{\cj} must be in the next period of \slot{q}{\ci}.
There must be a slot of \cua{}, \slot{p}{\cua}, such that $\slot{q}{\ci} < \slot{p}{\cua} < \slot{r}{\cj}$.
Consequently, there is a free entry in \slx{} in slot $\dst{p,\cua}$.
Because \cua{}'s request is not completed in slot \slot{t+T}{\cua}, it is not completed in slot $\dst{p,\cua}$.
The only reason that $\cua$'s request is not completed in its slot when there is a free entry is that \cua{} is performing a write-back, which contradicts the hypothesis that \cua{} does not perform any write-backs.

\end{proof}

\begin{corollary}\label{corollary:decreasing-3}
If \cua{} does not perform any write-backs, $l_x\in\slx$ is evicted in response to \cua{}'s request in \slot{t}{\cua} and \cua{}'s request is not completed in \slot{t+2(n-1)}{\cua} then
\begin{align*}
    \dst{\clx,\cua} (\slot{t+2(n-1)}{\cua}) < \dst{\clx,\cua}(\slot{t}{\cua})
\end{align*}

\end{corollary}
\begin{proof}
Assume that at $\stua{}$, cache line $l_x$ is evicted, but it is also privately cached by $c_i$ such that $\dst{\ci,\cua}=\dst{\clx,\cua}(\stua)$.
At a later slot for \ci{}, $\slot{q}{\ci}$ where $l_x$ is written back, $q\leq t+2(n-2) + 1 < t+2(n-1)$. This is because there can be at most $(n-1)$ pending write-backs in \ci{}'s PWB including the write-back for $l_x$. 
Before \cua{}'s next slot, another core \cj{} must occupy $l_x$ so that \cua{}'s request is not completed.
Due to \onetdm{}, if $\cj$ occupies $l_x$, $\dst{cj,\cua}<\dst{\ci,\cua}\leq\dst{\clx,\cua}(\stua)$.
Applying Lemma~\ref{lemma:decreasing-2}, $\slot{t+2(n-1)}{\cua}\leq\dst{cj,\cua}<\dst{\clx,\cua}(\stua)$.
\end{proof}


\begin{lemma}\label{lemma:write-back}
Given a slot \stua{} where \cua{} performs write-back,
then there exists an execution such that
\begin{align*}
    \forall l\in\slx{}: \dst{\cl,\cua} ( \slot{t+1}{\cua} )& \geqslant \dst{\cl{},\cua} (\slot{t}{\cua}).
\end{align*}
\end{lemma}

\begin{proof}
In \stua{}, let us assume that \cua{} writes back a cache line as a response to an eviction caused by another core. 
Since \cua{} is performing a write-back, it cannot issue a request; thus, its request cannot complete. 
Hence, for each of the free cache line entry $l$ in \slx{} at \stua{}, a core \cj{} can make a request to $l$ after \stua{} which completes within one slot.
Then, $\dst{\cl,\cua} (\slot{t+1}{\cua}) = \dst{\cj,\cua} >\dst{\cl,\cua} (\stua{})$.
For other cache lines $l' \in \slx{}$ that are privately cached by other cores that are not evicted due to accesses made by some other cores, $\dst{\cl,\cua} (\slot{t+1}{\cua}) = \dst{\cl,\cua} (\stua{})$ holds trivially.
\end{proof}

Corollary~\ref{corollary:decreasing-3} and lemma~\ref{lemma:write-back} provide the cornerstone to derive the worst-case latency for \cua{} in theorem~\ref{thm:bound-simple}. 
%
        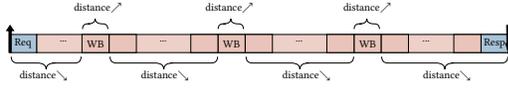
\begin{figure}[tb]
        \centering
        \resizebox{0.8\columnwidth}{!}{\tikzset{every picture/.style={line width=0.75pt}} 

\begin{tikzpicture}[x=0.75pt,y=0.75pt,yscale=-1,xscale=1]

\draw  [fill={rgb, 255:red, 232; green, 193; blue, 182 }  ,fill opacity=0.75 ] (160,270) -- (680,270) -- (680,290) -- (160,290) -- cycle ;
\draw  [fill={rgb, 255:red, 232; green, 193; blue, 182 }  ,fill opacity=0.75 ] (210,270) -- (240,270) -- (240,290) -- (210,290) -- cycle ;
\draw  [fill={rgb, 255:red, 169; green, 196; blue, 215 }  ,fill opacity=1 ] (130,270) -- (160,270) -- (160,290) -- (130,290) -- cycle ;
\draw  [fill={rgb, 255:red, 0; green, 0; blue, 0 }  ,fill opacity=1 ] (134.49,265.73) -- (130.75,260) -- (127,265.73) -- (129.5,265.73) -- (129.5,289.61) -- (131.99,289.61) -- (131.99,265.73) -- cycle ;
\draw  [fill={rgb, 255:red, 232; green, 193; blue, 182 }  ,fill opacity=0.75 ] (360,270) -- (390,270) -- (390,290) -- (360,290) -- cycle ;
\draw  [fill={rgb, 255:red, 232; green, 193; blue, 182 }  ,fill opacity=0.75 ] (510,270) -- (540,270) -- (540,290) -- (510,290) -- cycle ;
\draw  [fill={rgb, 255:red, 169; green, 196; blue, 215 }  ,fill opacity=1 ] (650,270) -- (680,270) -- (680,290) -- (650,290) -- cycle ;
\draw  [fill={rgb, 255:red, 0; green, 0; blue, 0 }  ,fill opacity=1 ] (677.34,279.96) -- (680,290) -- (682.34,279.89) -- (681.09,279.9) -- (680.78,259.98) -- (678.28,260.02) -- (678.59,279.94) -- cycle ;
\draw   (132,294.5) .. controls (132,299.17) and (134.33,301.5) .. (139,301.5) -- (160.5,301.5) .. controls (167.17,301.5) and (170.5,303.83) .. (170.5,308.5) .. controls (170.5,303.83) and (173.83,301.5) .. (180.5,301.5)(177.5,301.5) -- (202,301.5) .. controls (206.67,301.5) and (209,299.17) .. (209,294.5) ;
\draw   (538,264) .. controls (538,260.18) and (536.09,258.27) .. (532.27,258.27) -- (532.27,258.27) .. controls (526.81,258.27) and (524.08,256.36) .. (524.08,252.54) .. controls (524.08,256.36) and (521.35,258.27) .. (515.89,258.27)(518.35,258.27) -- (515.89,258.27) .. controls (512.07,258.27) and (510.16,260.18) .. (510.16,264) ;
\draw  [fill={rgb, 255:red, 232; green, 193; blue, 182 }  ,fill opacity=0.75 ] (240,270) -- (270,270) -- (270,290) -- (240,290) -- cycle ;
\draw  [fill={rgb, 255:red, 232; green, 193; blue, 182 }  ,fill opacity=0.75 ] (330,270) -- (360,270) -- (360,290) -- (330,290) -- cycle ;
\draw  [fill={rgb, 255:red, 232; green, 193; blue, 182 }  ,fill opacity=0.75 ] (390,270) -- (420,270) -- (420,290) -- (390,290) -- cycle ;
\draw  [fill={rgb, 255:red, 232; green, 193; blue, 182 }  ,fill opacity=0.75 ] (480,270) -- (510,270) -- (510,290) -- (480,290) -- cycle ;
\draw  [fill={rgb, 255:red, 232; green, 193; blue, 182 }  ,fill opacity=0.75 ] (540,270) -- (570,270) -- (570,290) -- (540,290) -- cycle ;
\draw  [fill={rgb, 255:red, 232; green, 193; blue, 182 }  ,fill opacity=0.75 ] (620,270) -- (650,270) -- (650,290) -- (620,290) -- cycle ;
\draw   (241,294.5) .. controls (241,299.17) and (243.33,301.5) .. (248,301.5) -- (290.25,301.5) .. controls (296.92,301.5) and (300.25,303.83) .. (300.25,308.5) .. controls (300.25,303.83) and (303.58,301.5) .. (310.25,301.5)(307.25,301.5) -- (352.5,301.5) .. controls (357.17,301.5) and (359.5,299.17) .. (359.5,294.5) ;
\draw   (391,294.5) .. controls (391,299.17) and (393.33,301.5) .. (398,301.5) -- (440.25,301.5) .. controls (446.92,301.5) and (450.25,303.83) .. (450.25,308.5) .. controls (450.25,303.83) and (453.58,301.5) .. (460.25,301.5)(457.25,301.5) -- (502.5,301.5) .. controls (507.17,301.5) and (509.5,299.17) .. (509.5,294.5) ;
\draw   (540.5,294.5) .. controls (540.5,299.17) and (542.83,301.5) .. (547.5,301.5) -- (600,301.5) .. controls (606.67,301.5) and (610,303.83) .. (610,308.5) .. controls (610,303.83) and (613.33,301.5) .. (620,301.5)(617,301.5) -- (672.5,301.5) .. controls (677.17,301.5) and (679.5,299.17) .. (679.5,294.5) ;
\draw   (388,264) .. controls (388,260.18) and (386.09,258.27) .. (382.27,258.27) -- (382.27,258.27) .. controls (376.81,258.27) and (374.08,256.36) .. (374.08,252.54) .. controls (374.08,256.36) and (371.35,258.27) .. (365.89,258.27)(368.35,258.27) -- (365.89,258.27) .. controls (362.07,258.27) and (360.16,260.18) .. (360.16,264) ;
\draw   (238.5,264) .. controls (238.5,260.18) and (236.59,258.27) .. (232.77,258.27) -- (232.77,258.27) .. controls (227.31,258.27) and (224.58,256.36) .. (224.58,252.54) .. controls (224.58,256.36) and (221.85,258.27) .. (216.39,258.27)(218.85,258.27) -- (216.39,258.27) .. controls (212.57,258.27) and (210.66,260.18) .. (210.66,264) ;

\draw (167.5,308.6) node [anchor=north] [inner sep=0.75pt]  [font=\normalsize] [align=left] {distance$\displaystyle \searrow $};
\draw (145,280) node   [align=left] {Req};
\draw (225,280) node  [font=\normalsize] [align=left] {WB};
\draw (375,280) node  [font=\normalsize] [align=left] {WB};
\draw (525,280) node  [font=\normalsize] [align=left] {WB};
\draw (665,280) node  [font=\normalsize] [align=left] {Resp};
\draw (189.51,277.5) node   [align=left] {...};
\draw (299.51,277.5) node   [align=left] {...};
\draw (449.51,277.5) node   [align=left] {...};
\draw (589.51,277.5) node   [align=left] {...};
\draw (524.5,247.4) node [anchor=south] [inner sep=0.75pt]  [font=\normalsize] [align=left] {distance$\displaystyle \nearrow $};
\draw (301.5,308.6) node [anchor=north] [inner sep=0.75pt]  [font=\normalsize] [align=left] {distance$\displaystyle \searrow $};
\draw (457.5,308.6) node [anchor=north] [inner sep=0.75pt]  [font=\normalsize] [align=left] {distance$\displaystyle \searrow $};
\draw (612.5,308.6) node [anchor=north] [inner sep=0.75pt]  [font=\normalsize] [align=left] {distance$\displaystyle \searrow $};
\draw (372.5,247.4) node [anchor=south] [inner sep=0.75pt]  [font=\normalsize] [align=left] {distance$\displaystyle \nearrow $};
\draw (227.5,247.4) node [anchor=south] [inner sep=0.75pt]  [font=\normalsize] [align=left] {distance$\displaystyle \nearrow $};

\end{tikzpicture}}
        \caption{An illustration that shows the WCL of \cua{}}
        \label{fig:simple-bound}
    \end{figure}
\begin{theorem}\label{thm:bound-simple}
Let $m=\min(\mcua, M)$, where \mcua{} is the cache capacity of \cua{}.
The worst-case latency in number of slots of the request of the core under analysis \cua{}, $WCL$, is given by:
\begin{equation}
    WCL =  \big((m + 1)\cdot A \cdot N + 1\big)\cdot{}\Sw,
\end{equation}
where $A = 2(n-1)\cdot w\cdot(n-1)$.
\begin{proof}
The critical instance has \cua{} making a request and receiving a response with the possibility of multiple write-backs from any core in between as shown in Figure~\ref{fig:simple-bound}. 
We split this critical instance into  four parts.
(1) The number of write-backs \cua{} can perform in the worst-case.
(2) The worst-case latency between two write-backs by \cua{}.
(3) The worst-case latency before the first write-back of \cua{}.
(4) The worst-case latency after the last write-back of \cua{} until it receives its response. 
    For (1), in the worst-case, other cores cause $m=\min(\mcua,M)$ write-backs on \cua{}, which is the maximal number of cache lines \cua{} can cache with partition $\mathcal{P}$.
    For (2), we showed that the \dd{} for a given cache line can both increase and decrease under certain situations.
    According to Lemma~\ref{lemma:write-back}, after a write-back, $\dst{\cl, \cua{}}$ can increase from $1$ to $n$ in the worst-case for all $l\in\slx{}$. Note that the distance would be $n$ if the core just after \cua{} was to privately cache the requested line.
    From Corollary~\ref{corollary:dist-range},  $\dst{\cl,\cua{}}$ ranges from $1$ to $n$ and in the worst-case, $\dst{\cl,\cua{}}$ can decrease from $n$ to $1$ for each of the $w$ cache lines $l\in\slx{}$ \emph{before} encountering the next write-back in the worst-case.
    Corollary~\ref{corollary:decreasing-3} shows that it takes $2(n-1)$ periods in the worst case to strictly decrease $\dst{\cl,\cua}$, and the worst-case decrement of \dd{} is by 1.
    Hence, for all $w$ cache lines to decrease from $n$ down to $1$, it takes $A = 2(n-1) \cdot{} w(n-1)$ periods, or $A \cdot N$ slots.
    For (3), in the worst-case, the \dd{} of all $w$ cache line entries in \slx{} decreases from $n$ down $1$ before the first write-back. 
    Following a similar argument as in (2), the WCL in (3) is hence $A$ periods or $A\cdot N$ slots.
    Similar to (3), for (4), after the last write-back, the \dd{} of all $w$ cache line entries in \slx{} decreases from $n$ down $1$ before \cua{} receives its response, and finally, one slot is required for \cua{} to receive its response, which translates to a worst-case latency of $A\cdot{}N + 1$ slots.
    Combining (1), (2), (3) and (4),
    $WCL = ((m-1)\cdot(A\cdot{}N) + (A\cdot{}N) + (A\cdot{}N+1))\cdot{}\Sw = ((m+1)\cdot{}A\cdot{}N+1)\cdot{}\Sw$


    
    
\end{proof}
\end{theorem}

\subsection{Set sequencer: Lowering the WCL}
We propose a hardware extension called a \stsqr{} that enables us to significantly lower the WCL.
Recall that the WCL analysis yields a WCL for a core under analysis \cua{} to be proportional to the minimum of either the cache capacity or \cua{}'s LLC partition size, and cube of the number of cores. 
This bound is clearly large making it difficult to allow cores to share partitions in the LLC. 
When using the \stsqr{}, the WCL bound ends up being \emph{independent} of the minimum of the cache capacity of \cua{} and \cua{}'s LLC partition size $M$.
We illustrate the main idea behind \stsqr{} using Figure~\ref{fig:ss}, which contains two structures, a \gls{qlt}~(\circled{1}) and a \gls{sqs}~(\circled{2}).
The \stsqr{} contains one entry in the \gls{qlt} for each set in the partition that has at least one pending LLC request. 
The entry maps the set to a queue in \gls{sqs}.
For example, $c_1$ has requested for set $3$, but $c_1$ has yet to occupy a free cache line in that set.
This may be because another core may still have to evict a cache line from their private caches before set $3$ has a free cache line.
When there are multiple cores requesting a cache line from the same set, such as set $5$, which maps to queue 2 in \gls{sqs}, then \stsqr{} stores the order in which the requests arrived at the LLC (broadcast order on the shared bus). 
For this set, core $c_2$ would occupy a free cache line in set $5$ before $c_3$, and so on. 
Our key observations and the WCL analysis revealed that the \dd{} increases whenever \cua{} is prevented from occupying a free cache line entry due to another core intercepting it. 
By maintaining order using \stsqr{}, we can guarantee that does not happen. 
\begin{figure}[bt]
    \centering
    \includegraphics[width=0.65\linewidth]{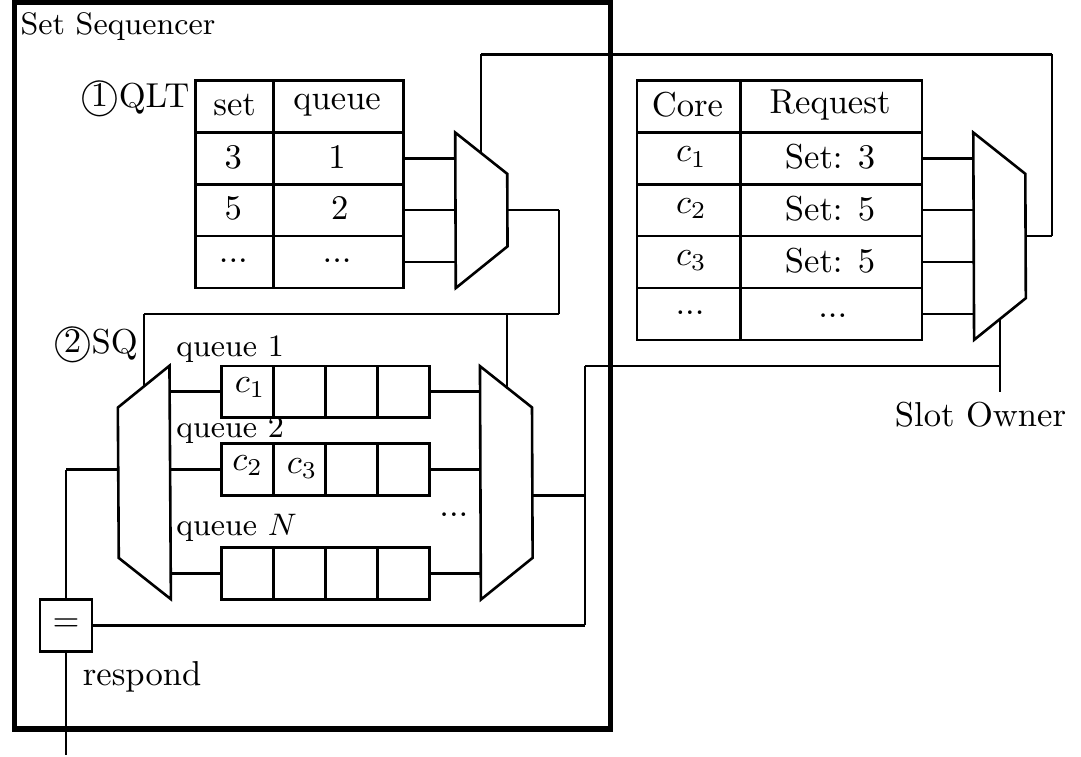}
    \caption{An illustration of \stsqr{}.}\label{fig:ss}
\end{figure}
\begin{figure}[b]
    \centering
    \includegraphics[width=0.80\linewidth]{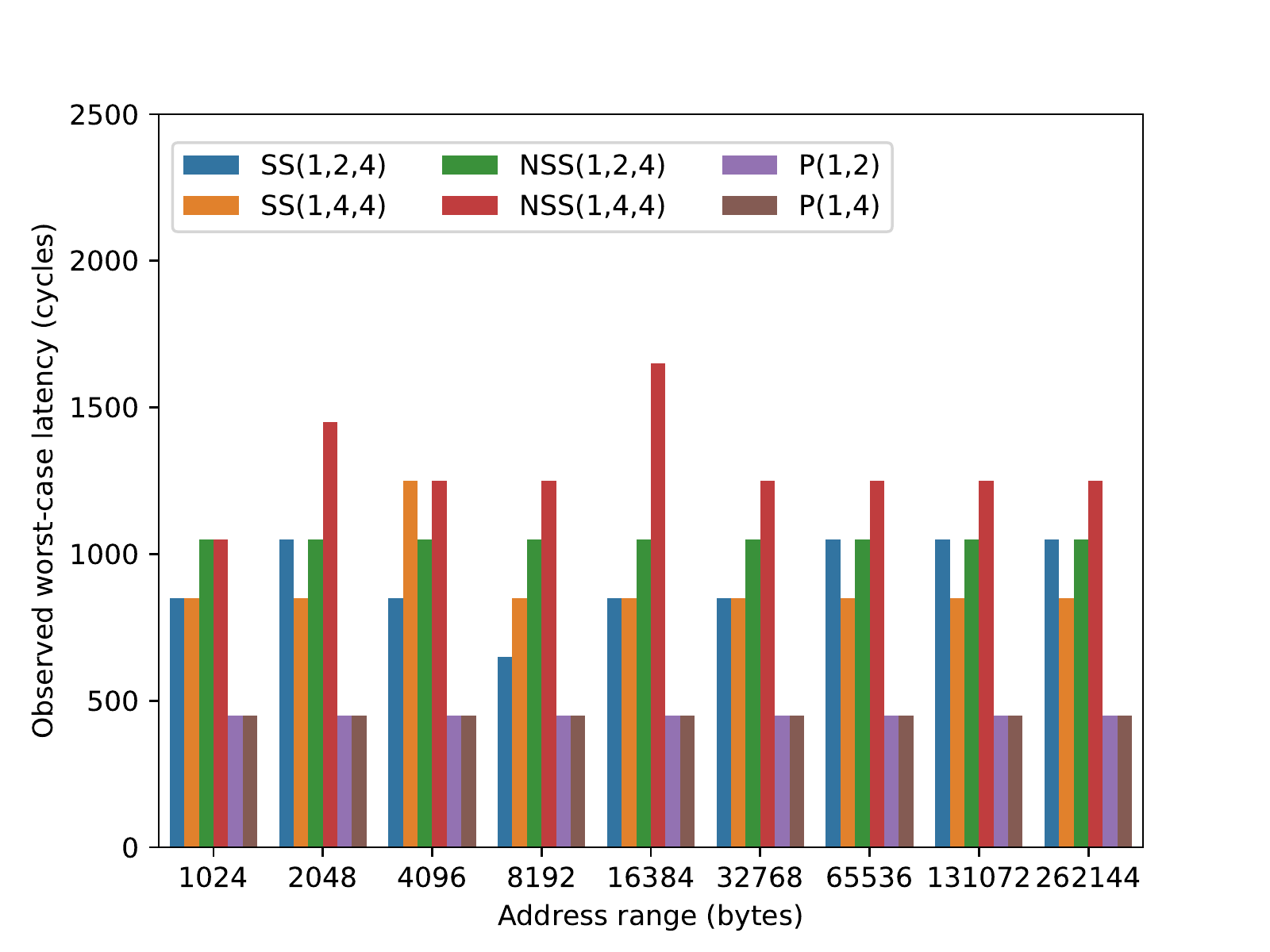}
    \caption{The observed WCL of \ssqx{}, \nssx{} and \ptnx{}.}
    \label{fig:eval-wc}
\end{figure}
\begin{figure*}[t]
    \begin{subfigure}{0.23\textwidth}
        \centering
        \includegraphics[width=\textwidth]{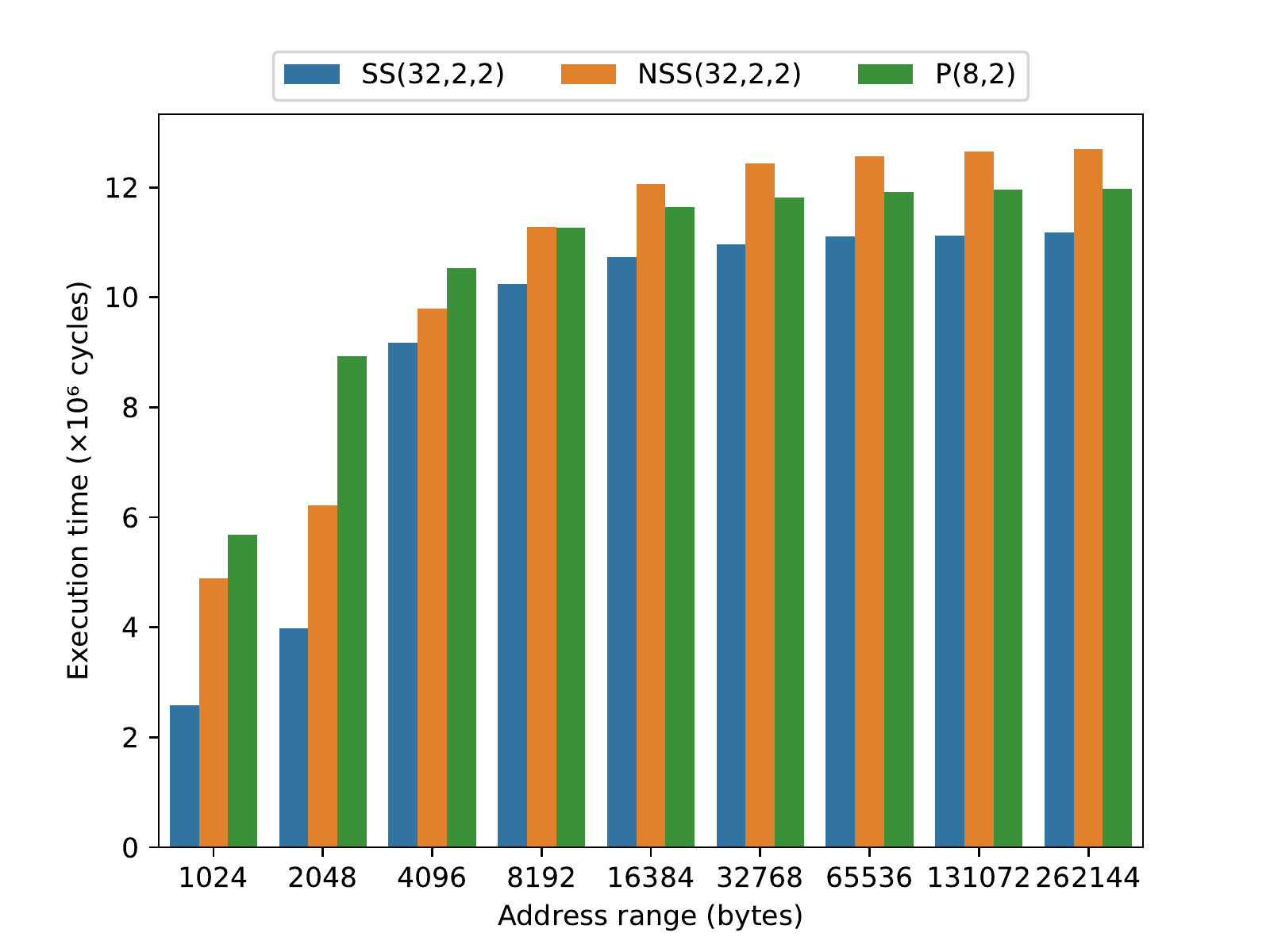}
        \caption{2-core, 4096B partition.}
        \label{fig:eval-large-time-2-core}
    \end{subfigure}
        \begin{subfigure}{0.23\textwidth}
        \centering
        \includegraphics[width=\textwidth]{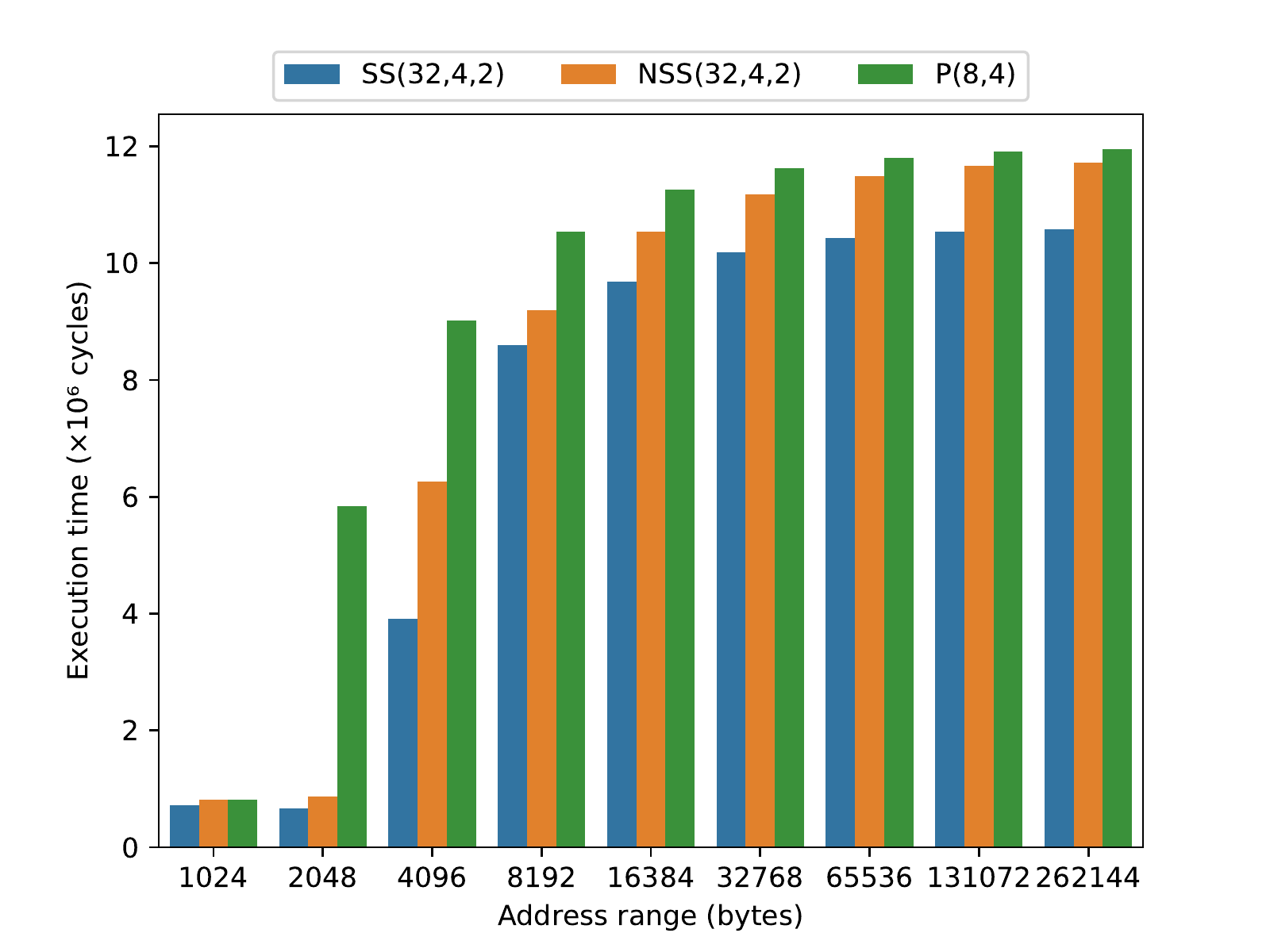}
        \caption{2-core, 8192B partition.}
        \label{fig:eval-large-time-2-core-4-way}
    \end{subfigure}
        \begin{subfigure}{0.23\textwidth}
         \centering
         \includegraphics[width=\textwidth]{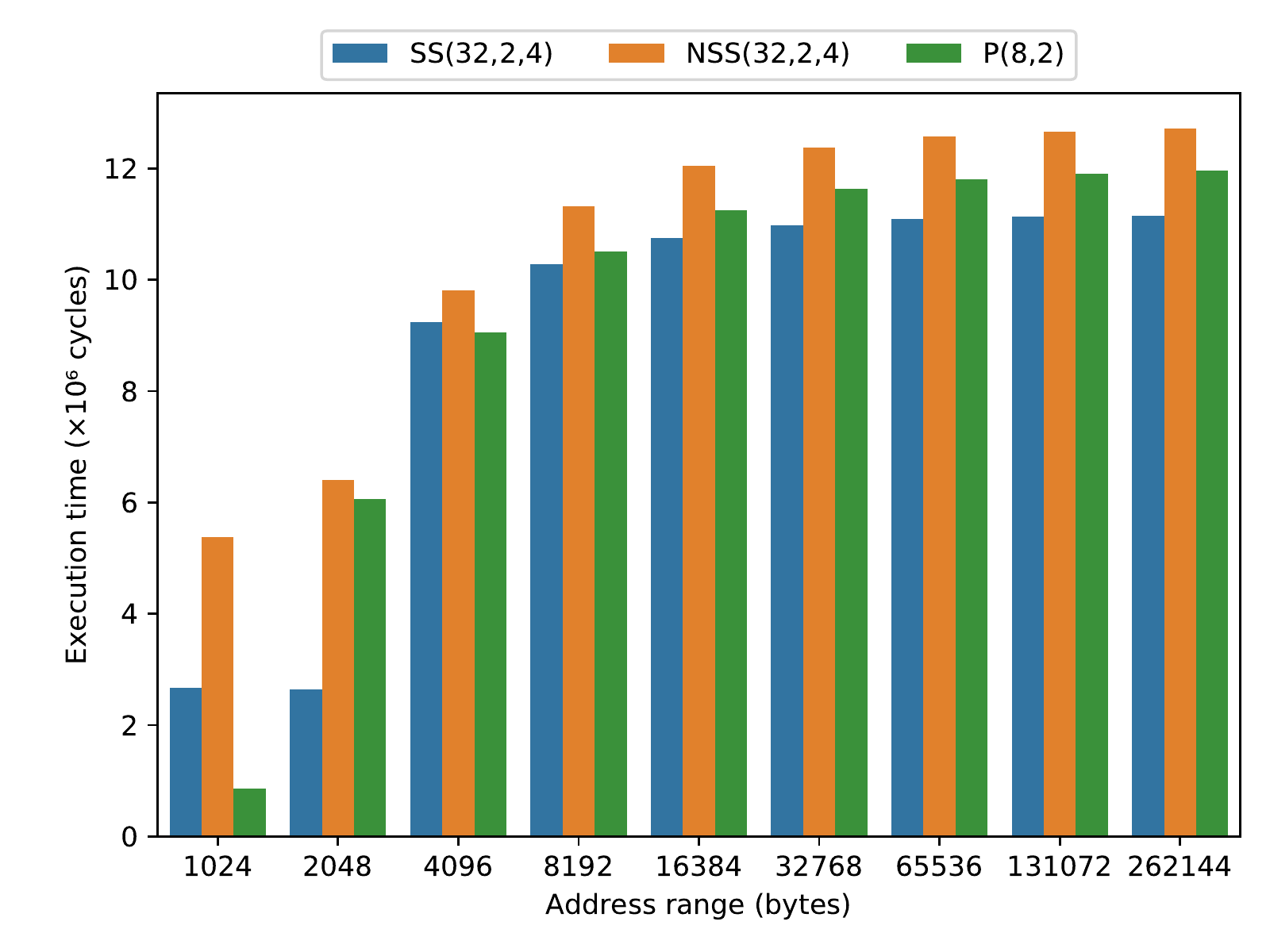}
        \caption{4-core, 4096B partition.}
         \label{fig:eval-large-time-4-set}
    \end{subfigure}
    \begin{subfigure}{0.23\textwidth}
        \centering 
        \includegraphics[width=\textwidth]{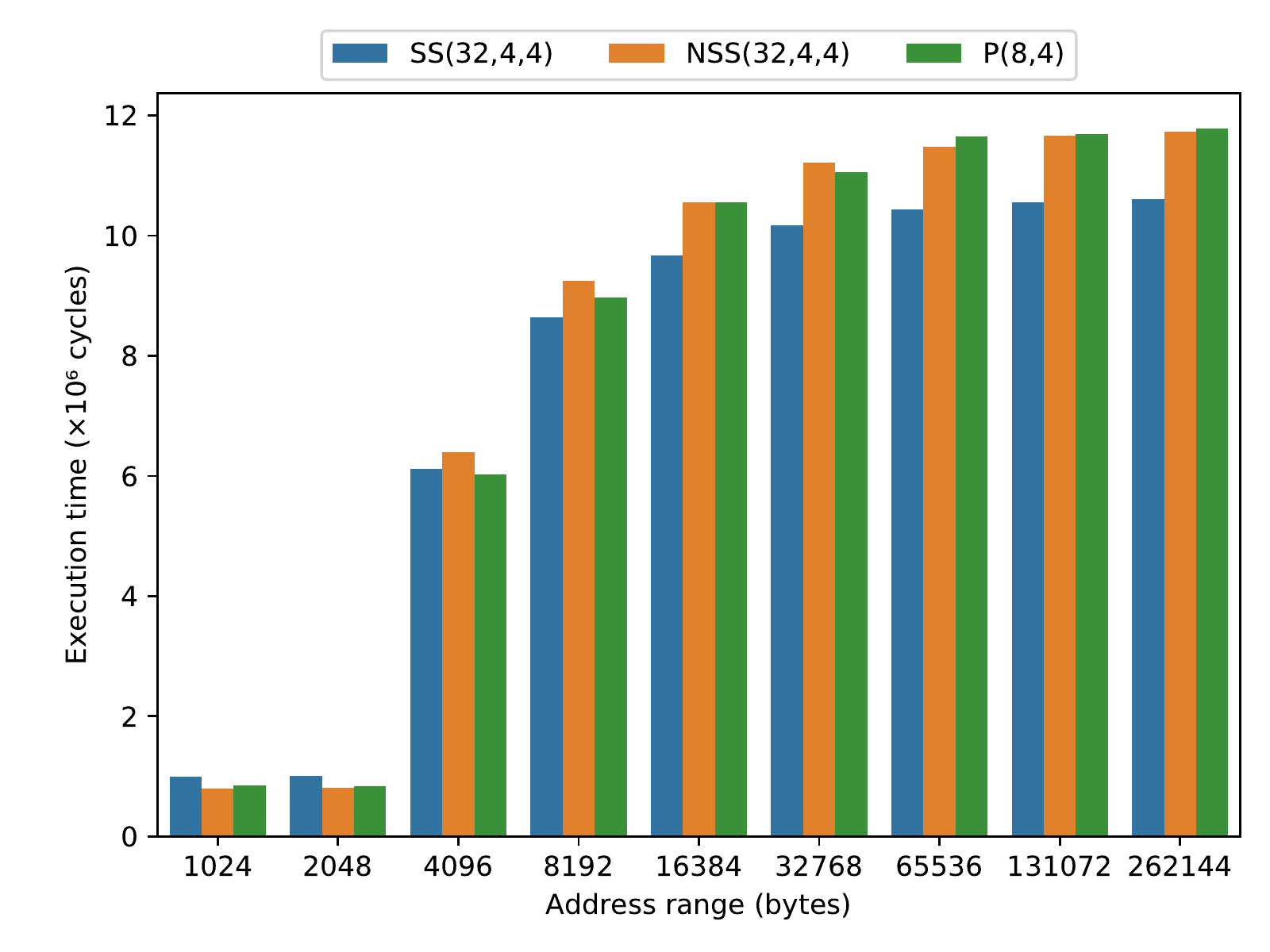}
        \caption{4-core, 8192B partition.}
        \label{fig:eval-large-time-32-set}
    \end{subfigure}
    \caption{The execution time of synthetic workload with fixed total partition size.}\label{fig:eval-large-same-size}
\end{figure*}
We present the WCL when using the \stsqr{} in Theorem~\ref{thm:better-bound}.
\begin{theorem}\label{thm:better-bound}
The WCL of a request of the core under analysis \cua{} when using the \stsqr{}, $WCL_{ss}$, is given by:
\begin{equation}
    WCL_{ss} =  \big(2(n-1)\cdot{}n + 1\big)\cdot{}N\cdot{}\Sw.
\end{equation}
\end{theorem}

\begin{proof}
In the worst-case, all other $(n-1)$ cores issue their request before \cua{} sends its request to cache line $X$ to the LLC, and it is the last request in the  \stsqr{} for a full set \slx{}.
For each request in the \stsqr, including \cua{}, it takes $2(n-1)$ slots for the core caching cache lines in \slx{} to write back a cache line and free an entry as a core performs $(n-1)$ write-backs in the worst-case, and the evicted cache line is written-back last.
Note that each such slot accounts for one period, which is $N\cdot{}\Sw$
Finally, \cua{} requires one slot \Sw{} to receive its response, which accounts for another period.
\end{proof}


\section{Evaluation}
Our empirical evaluation is performed with an in-house trace simulator that simulates the cache subsystem of a four-core system with the memory hierarchy as described in section~\ref{sec:sysmodel}.
The L2 cache is a 4-way set-associative cache with 16 sets and the L3 cache is a 16-way set-associative cache with 32 sets that can be partitioned across the four cores. 
The cache line size is 64-byte.

\noindent\textbf{Workload generation.}
We use synthetic workloads consisting of memory requests to random addresses within various address ranges.
We enforce disjoint address ranges for each core to guarantee that accesses to shared data does not occur.
For a certain address range, a core issues the same memory addresses across different partitioned configurations.

\noindent\textbf{Notation.} 
We use the following syntax to express partitioned configurations.
(1) \ssq{s,w,n}: a partition shared among $n$ cores with $s$ sets and $w$ ways with set sequencer, 
(2) \nss{s,w,n}: a partition shared among $n$ cores with $s$ sets and $w$ ways and LLC services contending requests with best effort, and
(3) \ptn{s,w}: a partition with $s$ sets and $w$ ways that is uniquely occupied by a core.
For \ptn{s,w}, each core is assigned equally-sized partition.


\subsection{Worst-case latency}
\noindent\textbf{Workload setup.}
To exercise the worst-case, we enforce a partition size of one set for all configurations.
This is done to force as many conflicts as possible. 
%
%



\noindent\textbf{Results.} Figure~\ref{fig:eval-wc} confirms that the observed WCL of all configurations are within the analytical WCLs, which are $5000$ cycles for \ssqx{}, $979250$ cycles for \nssx{}, and $450$ cycles for \ptnx{}.
Although a distinct partition \ptnx{} yields the lowest WCL, 
recall that we wish to share partitions for cores whose real-time requirements are met with sharing. 
However, there might be others that need distinct partitions \ptnx{}, which do indeed provide the lowest WCL.
This is essential when the number of required functionalities deployed onto a single multicore increases.
In the case of cores sharing a partition, the bound for \ssqx{} can be employed.
\nssx{} shows a higher observed WCL compared to \ssqx{} across all address ranges because \dd{} can increase as mentioned in Observation~3.

\subsection{Partition sharing and utilization}

We next investigate the the impact of partitioning when cores are forced to share a partition.

\noindent\textbf{Workload setup.}
We conduct the experiment with 2-core and 4-core setups, each with a fixed cache capacity that is then partitioned.
In \ssqx{} and \nssx{}, all cores share the same partition while in \ptnx{}, the fixed cache capacity is divided equally between all cores, and the set-associativity is fixed.
%
Figure~\ref{fig:eval-large-time-2-core} shows that when the address range is 1024-byte or 2048-byte, the execution time is the same across \ssqx{}, \nssx{} and \ptnx{}.
This is because the address range is less than or equal to the partition size. 

When the address range exceeds the partition size, \ssqx{} exhibits improved performance when compared to both \nssx{} and \ptnx{}.
In the 2-core setup with 4096-byte of partition size, \ssqx{} achieves an average speedup of $1.34\times$ as is shown in Figure~\ref{fig:eval-large-time-2-core}.
When the capacity is 8192-byte, \ssqx{} achieves an average speed up of $2.13\times$ (Figure~\ref{fig:eval-large-time-2-core-4-way}).
Such performance persists in the 4-core setup where \ssqx{} features an average speedup of $1.10\times$ for 4096-byte partition size (Figure~\ref{fig:eval-large-time-4-set}) and an average speed up of $1.02\times$ for 8192-byte partition size (Figure~\ref{fig:eval-large-time-32-set}).

\section{Conclusion}
This paper provides a complementary approach to strictly partitioning the LLC where cores can share LLC partitions. 
We expect sharing of partitions to be important as the demands for consolidating a large number of safety-critical functionalities onto a single multicore are accelerating.
Using a constrained TDM policy, multiple cores can predictably share the LLC.
However, the resultant WCL is grossly pessimistic. 
We introduced the \stsqr{} hardware structure that reduced the WCL by 2048 times, and empirically evaluated that the WCL bounds hold.

\bibliographystyle{ACM-Reference-Format}
\bibliography{refs-short}


\end{document}